\documentclass[11pt,letterpaper]{article}
\usepackage[lmargin=1.5in,rmargin=1.5in,bottom=1.3in,top=1.3in,
twoside=False]{geometry}

\usepackage{fullpage,amssymb,amsmath}
\usepackage{graphicx,algorithmic,algorithm}
\usepackage{enumerate}
\usepackage[T1]{fontenc}

\usepackage{tikz}
\usetikzlibrary{hobby}
\usetikzlibrary{arrows,shapes, decorations,snakes, decorations.pathreplacing}
\usetikzlibrary{arrows,calc,intersections}
\usepackage[fancy,color=orange]{tikz-inet}
\tikzstyle{vertex}=[circle,inner sep=2,minimum size =2mm,semithick,fill=white!80!blue, draw=black]
\tikzstyle{brace}=[decoration=brace,decorate]

 \usepackage{xcolor}
 \usepackage{mathtools}
\usepackage{microtype}
\usepackage{amsfonts}
\usepackage{comment}
\usepackage[english]{babel}
\usepackage{mathrsfs}

\usepackage{fontaxes}
\usepackage{trimspaces}
\usepackage{nccfoots}
\usepackage{setspace}

\usepackage{enumitem}
\usepackage{todonotes}

\renewcommand{\epsilon}{\varepsilon}

\newcommand{\Oof}{\mathcal{O}}
\newcommand{\CCC}{\mathcal{C}}
\newcommand{\FFF}{\mathcal{F}}
\newcommand{\PPP}{\mathcal{P}}

\newcommand{\ie}{i.e.\@ }
\newcommand{\sth}{\mathop{:}}
\newcommand{\N}{\mathbb{N}}
\renewcommand{\phi}{\varphi}

\newcommand{\minor}{\preccurlyeq}

\newcommand{\LA}{L\hspace{0.3mm}(\hspace{-0.2mm}A\hspace{0.2mm})}

\newcommand{\abs}[1]{\ensuremath{\left\lvert#1\right\rvert}}
\renewcommand{\mid}{:}

\definecolor{blue}{rgb}{0.1,0.2,0.5}
\definecolor{brown}{rgb}{0.6,0.6,0.2}
\usepackage[amsmath,thmmarks,hyperref]{ntheorem}
\usepackage{cleveref}

\newenvironment{claimproof}[1]{\par\noindent\textit{Proof.}\space#1}{\hfill $\dashv$}

\crefformat{page}{#2page~#1#3}%
\Crefformat{page}{#2Page~#1#3}%
\crefformat{equation}{#2(#1)#3}%
\Crefformat{equation}{#2(#1)#3}%
\crefformat{figure}{#2Figure~#1#3}%
\Crefformat{figure}{#2Figure~#1#3}%
\crefformat{section}{#2Section~#1#3}
\Crefformat{section}{#2Section~#1#3}
\crefformat{chapter}{#2Chapter~#1#3}
\Crefformat{chapter}{#2Chapter~#1#3}
\crefformat{chapter*}{#2Chapter~#1#3}
\Crefformat{chapter*}{#2Chapter~#1#3}
\crefformat{part}{#2Part~#1#3}
\Crefformat{part}{#2Part~#1#3}
\crefformat{enumi}{#2(#1)#3}
\Crefformat{enumi}{#2(#1)#3}

\usepackage{enumerate}

\usepackage{latexsym}

\theoremnumbering{arabic}
\theoremstyle{plain}
\theoremsymbol{}
\theorembodyfont{\itshape}
\theoremheaderfont{\normalfont\bfseries}
\theoremseparator{.}

\newtheorem{theorem}{Theorem}
\crefformat{theorem}{#2Theorem~#1#3}
\Crefformat{theorem}{#2Theorem~#1#3}

\newcommand{\newtheoremwithcrefformat}[2]{%
  \newtheorem{#1}[theorem]{#2}%
  \crefformat{#1}{##2\MakeUppercase#1~##1##3}%
  \Crefformat{#1}{##2\MakeUppercase#1~##1##3}%
}
\newcommand{\newseptheoremwithcrefformat}[2]{%
  \newtheorem{#1}{#2}%
  \crefformat{#1}{##2\MakeUppercase#1~##1##3}%
  \Crefformat{#1}{##2\MakeUppercase#1~##1##3}%
}

\newtheoremwithcrefformat{lemma}{Lemma}
\newtheoremwithcrefformat{proposition}{Proposition}
\newtheoremwithcrefformat{observation}{Observation}
\newtheoremwithcrefformat{conjecture}{Conjecture}
\newtheoremwithcrefformat{corollary}{Corollary}
\newseptheoremwithcrefformat{claim}{Claim}
\theorembodyfont{\upshape}
\newtheoremwithcrefformat{example}{Example}
\newtheoremwithcrefformat{remark}{Remark}
\newtheoremwithcrefformat{definition}{Definition}

\theoremstyle{nonumberplain}
\theoremheaderfont{\scshape}
\theorembodyfont{\normalfont}
\theoremsymbol{\ensuremath{\square}}
\newtheorem{proof}{Proof}

\theoremsymbol{\ensuremath{\lrcorner}}

\def\cqedsymbol{\ifmmode$\lrcorner$\else{\unskip\nobreak\hfil
\penalty50\hskip1em\null\nobreak\hfil$\lrcorner$
\parfillskip=0pt\finalhyphendemerits=0\endgraf}\fi}

\usepackage{cancel}

\hypersetup{
    colorlinks = true,
	linkcolor={blue}, 
	citecolor={brown}
}

\begin{document}
\title{Polynomial Kernels and Wideness Properties of\\
Nowhere Dense Graph Classes\thanks{This work has been supported by the
European Research Council (ERC) under the European Union's Horizon
2020 research and innovation programme (ERC Consolidator Grant
DISTRUCT, grant agreement No 648527).}}  
\author{Stephan Kreutzer\\Technische Universit\"at Berlin\\\texttt{stephan.kreutzer@tu-berlin.de} \and Roman Rabinovich\\Technische Universit\"at Berlin\\\texttt{roman.rabinovich@tu-berlin.de} \and Sebastian Siebertz\\Technische Universit\"at Berlin\\\texttt{sebastian.siebertz@tu-berlin.de}}

\date{}
\maketitle

\begin{abstract}
Nowhere dense classes of
  graphs~\cite{nevsetvril2010first,nevsetvril2011nowhere} are very
  general classes of uniformly sparse graphs with several seemingly
  unrelated characterisations. From an algorithmic perspective, a
  characterisation of these classes in terms of \emph{uniform
    quasi-wideness}, a concept originating in finite model theory, has
  proved to be particularly useful. Uniform quasi-wideness is used
  in many fpt-algorithms on nowhere dense classes. However, the
  existing constructions showing the equivalence of nowhere denseness
  and uniform quasi-wideness imply a non-elementary blow up in the
  parameter dependence of the fpt-algorithms, making them infeasible
  in practice.
As a first main result of this paper, we use tools from logic, in
particular from a sub-field of model theory known as stability
theory, to establish polynomial bounds for the equivalence of
nowhere denseness and uniform quasi-wideness. 

A powerful method in parameterized complexity theory is to compute a
problem kernel in a pre-computation step, that is, to reduce the input
instance in polynomial time to a sub-instance of size bounded in the parameter only
(independently of the input graph size). Our new tools allow us to
obtain for every fixed radius $r\in \N$ 
a polynomial kernel for the \mbox{distance-$r$} dominating set problem
on nowhere dense classes of graphs. This result is particularly
interesting, as it implies that for every class $\CCC$ of graphs which
is closed under taking subgraphs, the
distance-$r$ dominating set problem admits a kernel on $\CCC$ for every value of~$r$ if, and
only if, it already admits a polynomial kernel for every value of $r$ (under the standard assumption that
$\text{FPT}\neq\text{W[2]}$).
\end{abstract}


\section{Introduction}

Given a graph $G$ and an integer $k$, the \textsc{Dominating Set}
problem is to determine the existence of a subset $D\subseteq V(G)$ of
size at most $k$ such that every vertex $u\in V(G)$ is
\emph{dominated} by $D$, that is, if $u$ does not belong to $D$, then
it must have a neighbour in $D$. The \textsc{Dominating Set} problem,
parameterized by the size of the solution $k$, plays a central role in
parameterized complexity theory, it is arguably one of the most
important examples of a W[$2$]-complete problem, and hence considered
intractable from the point of view of parameterized complexity on general
graphs. A problem is fixed-parameter tractable on a class~$\CCC$ of
graphs parameterized by the solution size~$k$, if there is an
algorithm deciding whether a graph $G\in \CCC$ admits a solution of
size $k$ in time $f(k)\cdot \abs{V(G)}^c$, for some computable function
$f$ and constant $c$.

A particularly fruitful approach in parameterized complexity
theory and algorithmic graph structure theory is the study of hard
computational problems on restricted classes of inputs. This research
is based on the observation that many problems such as
\textsc{Dominating Set}, which are considered intractable in general,
can be solved efficiently on classes of graphs such as graphs of
bounded treewidth, planar graphs, or more generally, graph classes
excluding a fixed minor.

An important goal of this line of research is to identify the most
general classes of graphs on which a wide range of algorithmic
problems can be solved efficiently.  In this context, classes of
graphs excluding a fixed minor have been studied intensively. More
recently, even more general classes of graphs such as those excluding
a fixed topological minor have received increased attention.

A useful method in parameterized complexity is to compute a
problem kernel in a polynomial time pre-computation step, that is, to
reduce the input instance to a sub-instance of size bounded in the
parameter only (independently of the input graph size).  The first
important result of this type for the \textsc{Dominating Set}
problem by Alber et al.~\cite{alber2004polynomial} showed that there
exists a kernel of linear size for the problem on planar
graphs. Linear kernels were later found for bounded genus
graphs~\cite{bodlaender2009meta}, apex-minor-free
graphs~\cite{fomin2010bidimensionality}, $H$-minor-free
graphs~\cite{fomin2012linear}, and $H$-topological-minor-free
graphs~\cite{FominLST13}.  

The algorithmic results on these graph classes are in one way or
another based on topological arguments which can be derived from
structure theorems for the corresponding class. Most notable structure theorems in this
context are Robertson and Seymour's structure theorem for
$H$-minor-free graphs~\cite{robertson2003graph} or its extension to
$H$-topological-minor-free graphs by Grohe and
Marx~\cite{grohe2015structure}. A complete shift in
paradigm was initiated by Ne\v{s}et\v{r}il and Ossona de Mendez with
their ground-breaking work on bounded
expansion~\cite{nevsetvril2008grad,nevsetvril2008gradb} and nowhere
dense classes of
graphs~\cite{nevsetvril2010first,nevsetvril2011nowhere}.  On these
classes, which properly extend the aforementioned classes defined by
excluded (topological) minors, many topological
arguments are replaced by much more general density based arguments. 

\smallskip
Formally, a graph $H$ with $V(H)=\{v_1,\ldots, v_n\}$ is a
\emph{minor} of a graph $G$, written $H\minor G$, if there are
pairwise vertex disjoint connected subgraphs $H_1,\ldots, H_n$ of~$G$
such that whenever $\{v_i,v_j\}\in E(H)$, then there are $u_i\in V(H_i)$
and $u_j\in (H_j)$ with $\{u_i,u_j\}\in E(G)$. We call
$(H_1,\ldots, H_n)$ a {\em{minor model}} of~$H$ in~$G$. The graph $H$
is a {\em{depth-$r$ minor}} of $G$, denoted $H\minor_rG$, if there is
a minor model $(H_1,\ldots,H_n)$ of~$H$ in $G$ such that each $H_i$
has radius at most $r$. We denote the complete graph on~$t$ vertices
by $K_t$ and the complete bipartite graph with parts of size $s$ and $t$
by $K_{s,t}$. 

\begin{definition}\label{def:nwd}
  A class $\CCC$ of graphs is \emph{nowhere dense} if there is a
  function $f\colon\N\rightarrow \N$ such that $K_{f(r)}\not\minor_r G$ for
  all $r\in \N$ and all $G\in \CCC$.
\end{definition}

It turned out that nowhere dense classes have many equivalent and
seemingly unrelated characterisations making it an extremely robust
and natural concept~\cite{NesetrilOdM12}.

Algorithmically, a characterisation of nowhere dense classes in terms
of \emph{uniform quasi-wideness}, a concept emerging from finite model
theory~\cite{dawar2010homomorphism}, has proved to be extremely
useful. A set $B\subseteq V(G)$ is $r$-independent in a graph $G$ if 
any two distinct vertices of $B$ have distance greater than $r$ in $G$. 

\begin{definition}\label{def:uqw}
  A class $\CCC$ of graphs is \emph{uniformly quasi-wide} if there are
  functions $N\colon\N\times\N\rightarrow \N$ and $s\colon\N\rightarrow \N$ such
  that for all $r,m\in \N$ and all subsets $A\subseteq V(G)$ for
  $G\in \CCC$ of size $\abs{A}\geq N(r,m)$ there is a set
  $S\subseteq V(G)$ of size $\abs{S}\leq s(r)$ and a set $B\subseteq A$ of
  size $\abs{B}\geq m$ which is $r$-independent in $G-S$. 
  The functions $N$ and
  $s$ are called the \emph{margin} of the class~$\CCC$.
\end{definition}

\begin{theorem}[Ne\v{s}et\v{r}il and Ossona de Mendez \cite{nevsetvril2010first}]
  A class $\CCC$ of graphs is nowhere dense if, and only if, it is uniformly
  quasi-wide.
\end{theorem}

The first fixed-parameter algorithms for the \textsc{dominating set} problem
on nowhere dense classes of graphs appeared in~\cite{DawarK09}.  As observed
in~\cite{DawarK09}, uniform quasi-wideness can be made algorithmic in the sense
that the sets $S$ and $B$ can be computed in polynomial time. This can be used
to define bounded search tree algorithms for problems such as \textsc{Dominating
  Set} parameterized by the solution size $k$ as follows. As long as the set $A$
of non-dominated vertices is large enough we are guaranteed to find a
$2$-independent subset $B$ of $A$ of size $k+1$ in $G$ once we removed a constant
size set $S$ of vertices from $G$. As no two elements of $B$ can be dominated by
a single element in $G-S$, it follows that the dominating set must contain an
element of the constant size set $S$. Trying every subset of $S$ as a part of
the dominating set and iterating this procedure until the number of
non-dominated vertices is bounded by a function of the parameter yields a
natural reduction. On the resulting structure one obtains the answer by 
brute force. With a little more effort this technique can be used to
establish fixed-parameter algorithms for many other problems,
see~\cite{DawarK09} for details.

A much more general result was achieved in~\cite{grohe2014deciding}.  Grohe et
al.\@ proved a very general algorithmic meta-theorem stating that first-order
model-checking is fixed parameter tractable on nowhere dense classes of graphs
(with the size of the formula as the parameter). This implies that a very
broad and natural class of algorithmic problems is fixed-parameter tractable on
nowhere dense classes of graphs. Again this proof uses uniform
quasi-wideness in its construction.

More recently nowhere dense classes of graphs have also been
studied in the context of kernelisation.
In~\cite{drange2016kernelization}, it was shown that
\textsc{Dominating Set} and \textsc{Distance-$r$ Dominating Set} admit a linear
kernel on bounded expansion classes and that \textsc{Dominating Set}
admits an almost linear kernel on nowhere dense classes of graphs. A
\emph{distance-$r$ dominating set} is a set $D\subseteq V(G)$ such
that every vertex $u\in V(G)$ has distance at most $r$ to a vertex
from $D$.  However, the techniques used
in~\cite{drange2016kernelization} are not strong enough to show that
also the \textsc{Distance-$r$ Dominating Set} problem admits a polynomial
kernel on nowhere dense classes of graphs. It was shown, however, that
for every class $\CCC$ of graphs which is closed under taking
subgraphs, if $\CCC$ admits a kernel for the \textsc{Distance-$r$ Dominating
  Set} problem
for every value of $r\in \N$, then $\CCC$ must be nowhere dense (under
the assumption $\text{W[$2$]}\neq\text{FPT}$).  These results were
complemented by lower bounds for the closely related \textsc{Connected
  Dominating Set} problem, where we are looking for a dominating set
$D$ which additionally must be connected. It was shown that there
exists a class of bounded expansion which is closed under taking
subgraphs that does not admit a polynomial kernel for
\textsc{Connected Dominating Set} (unless
$\text{NP}\subseteq \text{coNP/poly}$).

\paragraph{Our contributions.} From an algorithmic perspective, the
main problem with the characterisation of nowhere dense classes by
uniform quasi-wideness is that the functions $N$ and $s$ used in the
definition of the class are established by iterated Ramsey arguments
(see \cite{nevsetvril2010first,NesetrilOdM12}).  Therefore, the
function $N$ grows extremely fast and depends non-elementarily on the
size of the excluded cliques in the definition of nowhere
denseness. It follows that fixed-parameter algorithms using uniform
quasi-wideness, such as the algorithms in \cite{DawarK09} mentioned
above, have astronomical parameter dependence making them infeasible
in practice even for very small parameter values.

 Our first main result is to improve the bounds on
uniform quasi-wideness dramatically. In fact, we can show that a class~$\CCC$ is
uniformly quasi-wide with margin $N\colon \N\times \N\rightarrow \N$ if, and only
if, for every $r\in \N$ there is a polynomial $p_r(x)$ such that $\CCC$ is
uniformly quasi-wide with margin $N'(r,m)\leq p_r(m)$.  This is a direct
corollary of the following theorem, which we prove in~\cref{sec:uqw}.

\pagebreak

\newcounter{main_thm}
\setcounter{main_thm}{\value{theorem}}
\newcounter{main_thm_section}
\setcounter{main_thm_section}{\value{section}}
\begin{theorem}\label{thm:uqw}
  Let $\CCC$ be a nowhere dense class of graphs. For every $r\in \N$
  there exists a polynomial~$p_r(x)$ and a constant $s(r)$ such that
  for all $m\in \N$ the following holds. For all $G\in\CCC$ and all 
  sets $A\subseteq V(G)$ of size at least
  $p_r(m)$, there is a set $S\subseteq V(G)$ of size at
  most $s(r)$ such that there is a set $B\subseteq A$ of size at
  least~$m$ which is $r$-independent in $G-S$.

  Furthermore, if $K_c\not\minor_rG$ for all $G\in \CCC$, then
  $s(r)\leq c$ and there is an algorithm, that given an $n$-vertex graph
  $G\in \CCC$, $\epsilon>0$, $r\in \N$ and $A\subseteq V(G)$ of size at least
  $p_r(m)$, computes a set $S$ of size at most $s(r)$ and an
  $r$-independent set $B\subseteq A$ in $G-S$ of size at least $m$ in
  time $\Oof(r\cdot c\cdot |A|^{c+1}\cdot n^{1+\epsilon})$.
\end{theorem}

\begin{corollary}
  A class $\CCC$ of graphs is uniformly quasi-wide with margins
  $N\colon \N\times \N\rightarrow \N$ and $s\colon\N\rightarrow \N$ 
  if, and only if, it is uniformly
  quasi-wide with a polynomial margin 
  $N'(r,m)\leq p_r(m)$ and a 
  margin $s'(r)\leq c$ for a
  polynomial~$p_r(x)$ and a constant $c$ 
  depending on $r$ and $\CCC$ only.
\end{corollary}
 
Compare this result to a result of a similar flavour by Demaine and
Hajiaghayi~\cite{demaine2004equivalence}, stating that a minor closed
class $\CCC$ of graphs has bounded local treewidth if, and only if,
$\CCC$ has linearly bounded local treewidth.

\bigskip The polynomial bounds on the margin of uniformly quasi-wide
classes, and hence nowhere dense classes, give us a new tool to prove
polynomial kernels. As our second main algorithmic result, proved
in~\cref{sec:kernel}, we take a step towards solving an open problem
stated in \cite{drange2016kernelization}, to find an (almost) linear kernel for
the \textsc{Distance-$r$ Dominating Set} problem on nowhere dense
classes of graphs.

\begin{theorem}\label{thm:poly-kernel}
  For every fixed value $r\in \N$, there is a polynomial kernel for the
  \textsc{Distance-$r$ Dominating Set} problem on every nowhere dense class of
  graphs.
\end{theorem}

We remark that in \cite{drange2015kernelization} it was already shown that for
classes $\CCC$ that are closed under taking subgraphs, if $\CCC$ admits a kernel
for the \textsc{Distance-$r$ Dominating Set} problem for every $r\in \N$,
then~$\CCC$ must be nowhere dense (under the standard assumption that FPT $\neq$
W[2]). Hence, under this assumption, the theorem implies that a class $\CCC$
which is closed under taking subgraphs admits a kernel for the
\textsc{Distance-$r$ Dominating Set} problem for every $r\in\N$ if, and only if,
it admits a polynomial kernel for every $r\in\N$.

\bigskip As another consequence of~\cref{thm:uqw} we can
dramatically improve the parameter dependence of the
\textsc{Connected Dominating Set} problem. See~\cref{sec:single} for a
proof of the following theorem.

\begin{theorem}
  Let $\CCC$ be a nowhere dense class of graphs. Then there is a
  polynomial $p(x)$ and an algorithm running in time
  $2^{p(k)}\cdot n^{1+\epsilon}$ which, given an $n$-vertex graph $G$,
  $\epsilon>0$ and a
  number~$k$ as input, decides whether $G$ contains a connected
  dominating set of size~$k$.
\end{theorem}

We believe that \cref{thm:uqw} can allow to reduce the parameter
dependence for further fixed-parameter
algorithms on nowhere dense classes such as the algorithms developed
in~\cite{DawarK09} .

We prove our results using tools from a branch of model theory known
as \emph{stability theory}. Stability was introduced by Shelah as a
notion of well-behaved first-order logic theories. In a recent paper,
Malliaris and Shelah \cite{malliaris2014regularity} used the model
theoretic tools to study classes of \emph{stable graphs}. Stable
classes of graphs are much more general than nowhere dense classes of
graphs, but the two concepts coincide on classes of graphs closed
under taking subgraphs~\cite{adler2014interpreting}. The focus of \cite{malliaris2014regularity} is
on proving very strong regularity lemmas for stable graphs. For
obtaining these lemmas they prove a very nice technical lemma,
\cref{thm:extract_indiscernibles} below, on the existence of long
$\Delta$-indiscernible sequences in stable graphs, where $\Delta$ is a
finite set of first-order formulas.  We are going to use
$\Delta$-indiscernible sequences to extract large $r$-independent sets
for properly defined formula sets~$\Delta$.
See~\cref{sec:prelims} for the definition of $\Delta$-indiscernible
sequences.  One of the technical tools we develop in this paper is to
make this lemma of \cite{malliaris2014regularity} algorithmic so that
we can apply it in our algorithms.

We believe that stable classes will be very interesting for future
algorithmic research and may be a good candidate for a generalisation
of nowhere dense classes with good algorithmic properties towards
classes of graphs which are no longer closed under taking subgraphs
(but e.g.\@ are only closed under taking induced subgraphs). Our
technical results here may therefore be of independent interest as a
first step towards understanding the algorithmic context of stable
classes of graphs.

\section{Stability and Indiscernibles}
\label{sec:prelims}

\paragraph{Graphs.} We use standard graph theoretical notation
and refer to~\cite{diestel2012graph} for reference.  Let $G$ be an
undirected graph. We write $N(v)$ for the set of neighbours of a
vertex $v\in V(G)$ and $v$ itself and~$N_r(v)$ for the set of vertices at distance at
most~$r$ from~$v$, again including $v$.  A set $W\subseteq V(G)$ is called
\emph{$r$-independent} in $G$, if all distinct $u,v\in W$ have
distance greater than $r$ in $G$.

\paragraph{Ladder index, branching index and VC-dimension.}
Let $G$ be a directed graph. The \emph{ladder index} of~$G$ is the
largest number $k$ such that there are
$v_1,\ldots, v_k,w_1,\ldots, w_k\in V(G)$ with
\[(v_i, w_j)\in E(G)\;\Leftrightarrow\; i\leq j.\]

The figure shows a graph with ladder index $6$ (imagine the 
edges to be directed from $v_i$ to $w_j$). 
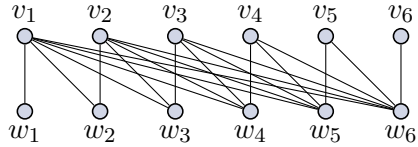
\bigskip
\begin{center}
\begin{tikzpicture}

\foreach \i in {1,...,6}{
  
  \node at (\i,1.3) {$v_\i$};
  \node[vertex] (a\i) at (\i,1){};
  \node[vertex] (b\i) at (\i,0){};
  \node at (\i,-0.3) {$w_\i$};
  
   \foreach \j in {1,...,\i}{
     \draw (b\i) to (a\j);
   }
}
\end{tikzpicture}
\end{center}
\bigskip

If $\tau$ is a word over an alphabet $\Sigma$ and
$a\in \Sigma$, then $\tau\cdot a$ denotes the concatenation of~$\tau$
and $a$.  The \emph{branching index} of $G$ is the largest number
$\ell$ such that there are vertices
$u_{\sigma_1},\ldots, u_{\sigma_{2^\ell}}\in V(G)$, indexed by the
words over the alphabet $\{0,1\}$ of length exactly $\ell$, and
vertices $v_{\tau_1},\ldots, v_{\tau_{2^\ell-1}}$, indexed by the
words over $\{0,1\}$ of length strictly smaller than $\ell$, such that
if $\tau_j\cdot a$ is a (not necessarily proper) prefix of~$\sigma_i$, then
$(u_{\sigma_i},v_{\tau_j})\in E(G)$ if, and only if, $a=1$. The vertices 
$u_{\sigma_1},\ldots, u_{\sigma_{2^\ell}}\in V(G)$ are called the 
\emph{leaves} of the tree, the vertices $v_{\tau_1},\ldots, v_{\tau_{2^\ell-1}}$
are its \emph{inner nodes}. Note that we describe edges only between inner
nodes and leaves. Intuitively, a leaf $u$ is connected to its predecessors~$v$ such
that $u$ is a \emph{right successor} of $v$ and not to its predecessors such that 
it is a \emph{left
successor}, while we make no assumptions on edges between different 
branches.  We call the graph
induced by vertices $u_{\sigma_1},\ldots, u_{\sigma_{2^\ell}}$ and
$v_{\tau_1},\ldots, v_{\tau_{2^\ell-1}}$ a \emph{branching witness} for~$G$ of
index~$\ell$.

\bigskip
The ladder index and branching index of $G$ are closely related, as
shown by the next lemma.

\begin{lemma}[\cite{hodges1993model}, Lemma 6.7.9, p.\
  313]\label{lem:branching}
  Let $G$ be a directed graph.  If $G$ has branching index~$k$, 
  then~$G$ has ladder index smaller than $2^{k+1}$.  If $G$ has
  ladder index $k$, then $G$ has branching index smaller than
  $2^{k+2}-2$.
 \end{lemma}

 We come to the definition of VC-dimension.  Let $A$ be a set and let
 $\FFF\subseteq \PPP ow(A)$ be a family of subsets of $A$. For a set
 $X\subseteq A$ let
\[X\cap \FFF\coloneqq\{X\cap F : F\in \FFF\}.\] 

The set $X$ is \emph{shattered by $\FFF$} if
\[X\cap \FFF=\PPP ow(X).\] 

The \emph{Vapnik-Chervonenkis-dimension}, short \emph{VC-dimension},
of $\FFF$ is the maximum size of a set $X$ that is shattered by
$\FFF$. Note that if $X$ is shattered by $\FFF$, then every
subset of $X$ is shattered by~$\FFF$.

The following theorem was first proved by Vapnik and
Chervonenkis~\cite{vapnik2015uniform}, and independently by
Sauer~\cite{sauer1972density} and
Shelah~\cite{shelah1972combinatorial}. It is often called the
Sauer-Shelah-Lemma in the literature.

\begin{theorem}[Sauer-Shelah-Lemma]\label{thm:sauer_shelah}
  If $\abs{A}\leq m$ and $\FFF \subseteq \PPP ow(A)$ has VC-dimension $d$,
  then
  \[\abs{\FFF}\leq \sum_{i=0}^{d}\binom{m}{i}\leq m^d+1.\]
\end{theorem}

The VC-dimension of an undirected graph $G$ is the VC-dimension of the
family of sets \[\FFF=\{N(v) : v\in V(G)\}.\]

\bigskip
\paragraph{First-order logic.}
For extensive background on first-order logic, we refer the reader
to~\cite{hodges1993model}. For our purpose, it suffices to define
first-order logic over the vocabulary of graphs (with constant symbols
from a given parameter set).
 
Let $A$ be a set. We call $\LA\coloneqq\{E\hspace{0.3mm}\}\mathop{\cup} A$ the \emph{vocabulary}
of graphs with parameters from $A$. \emph{First-order formulas} over $\LA$ are
formed from atomic formulas~$x=y$ and $E(x,y)$, where $x,y$ are variables (we
assume that we have an infinite supply of variables) or elements of $A$ treated
as constant symbols, by the usual Boolean
connectives~$\neg$~(negation),~$\wedge$ (conjunction), and~$\vee$ (disjunction)
and existential and universal quantification~$\exists x,\forall x$,
respectively.  The free variables of a formula are those not in the scope of a
quantifier, and we write~$\phi(x_1,\ldots,x_k)$ to indicate that the free
variables of the formula~$\phi$ are among $x_1,\ldots,x_k$. We often
abbreviate a tuple $(x_1,\ldots, x_k)$ by $\bar{x}$ and let
the context determine the length of a tuple $\bar{x}$. 

To define the semantics, we inductively define a satisfaction
relation~$\models$. Let $G$ be a graph and $A\subseteq V(G)$. For an
$\LA$-formula~$\phi(x_1,\ldots,x_k)$, and
$v_1,\ldots,v_k\in V(G)$, $G\models\phi(v_1,\ldots,v_k)$
means that~$G$ satisfies~$\phi$ if the free variables~$x_1,\ldots,x_k$
are interpreted by~$v_1,\ldots,v_k$ and the parameters $a\in A$
(formally treated as constant symbols) used in the formula are
interpreted by the corresponding element of $A$ in $G$, respectively. If
$\phi(x_1,x_2)=E(x_1,x_2)$ is atomic, then $G\models\phi(v_1,v_2)$
if~$(v_1,v_2)\in E(G)$. The meaning of the equality symbol, the
Boolean connectives, and the quantifiers is as expected. For a
formula $\phi(x_1,\ldots, x_k, y_1,\ldots, y_\ell)$ and
$v_1,\ldots, v_\ell\in V(G)$ (treated as a sequence of parameters), we
write $\phi(x_1,\ldots, x_k, v_1,\ldots, v_\ell)$ for the formula with
free variables $x_1,\ldots, x_k$ where each occurrence of the variable
$y_i$ in $\phi$ is replaced by the constant symbol $v_i$.

Let $\Delta$ be a set of formulas, let $G$ be a graph and let
$A\subseteq V(G)$. Let $v\in V(G)$. The \emph{$\Delta$-type of
  vertex~$v$ in $G$ over the parameters $A$} is the set
\begin{align*}
  \mathrm{tp}_\Delta(G, A, v) & \coloneqq  \{ \phi(x_1,v_1,\ldots, v_k) :
                                \phi(x_1,y_1,\ldots, y_k)\in \Delta,
                                v_1,\ldots, v_k\in A,
                                G\models\phi(v,v_1,\ldots, v_k)\}
  \\
                              & \hspace{1pt}\ \cup  \{ \neg \phi(x_1,v_1,\ldots, v_k) : \phi(x_1,y_1,\ldots, y_k)\in
                                \Delta, v_1,\ldots, v_k\in A, G\not\models\phi(v,v_1,\ldots, v_k)\}.
\end{align*}
The set of \emph{$\Delta$-types realised} in $G$ over $A$ is the set
$S_\Delta(G,A) \coloneqq \{ \mathrm{tp}_\Delta(G, A, v) \sth v\in V(G) \}$.

\begin{example}\label{ex:types}
  Let $\Delta^r$ be the set consisting of the single formula
  $\phi_r(x_1,y_1)$, stating that the elements~$x_1$ and $y_1$ have
  distance at most $r$ in a graph $G$. Let $G$ be a graph and let
  $A\subseteq V(G)$.
  \begin{enumerate}
  \item For $v\in V(G)$, we can identify
    $\mathrm{tp}_{\Delta^r}(G,A,v)$ with $N_r(v)\cap A$, in the sense
    that $a\in N_r(v)\cap A$ if, and only if, the formula
    $\phi_r(x_1, a)\in \mathrm{tp}_{\Delta^r}(G,A,v)$.
  \item If $G$ is an arbitrary graph, then we can have
    $\abs{S_{\Delta^1}(G,A)}=2^{\abs{A}}$.  If $G$ has VC-dimension~$k$, then
    $\abs{S_{\Delta^1}(G,A)}\leq \abs{A}^k+1$ according to the
    Sauer-Shelah-Lemma. If $G$ comes from a nowhere dense class of
    graphs and $\epsilon>0$, then we
    have $\abs{S_{\Delta^1}(G,A)}\leq \abs{A}^{1+\epsilon}$ for all sufficiently 
    large~$A$. This follows from Lemma 4.11(2)
    of~\cite{gajarsky2017kernelization}.
  \end{enumerate}
\end{example}

Using tools from stability theory, a much more general result can be
established (\cref{thm:shelah1} below).

\bigskip
\paragraph{First-order interpretations and stability.}
Let $\phi(x_1,\ldots, x_k)$ with $k\geq 2$ be a first-order formula
and let~$G$ be a graph.  For every \emph{ordered partition}
$(x_{i_1},\ldots, x_{i_\ell}), (x_{i_{\ell+1}},\ldots, x_{i_k})$ of
the variables $x_1,\ldots, x_k$ we define a directed graph
$G_{\phi((x_{i_1},\ldots, x_{i_\ell}), (x_{i_{\ell+1}},\ldots,
  x_{i_k}))}$
which has as its vertices the \mbox{$\ell$-tuples} $(v_1,\ldots, v_\ell)$ 
and $k-\ell$-tuples $(v_{\ell+1},\ldots, v_k)$ of
vertices of $G$ and all arcs
$((v_1,\ldots, v_\ell), (v_{\ell+1},\ldots, v_k))$ such that
$G\models\phi(v_{i_1},\ldots, v_{i_k})$. If the variable partition is
of relevance, we will always denote the formula as
$\phi((x_{i_1},\ldots, x_{i_\ell}),$ $(x_{i_{\ell+1}},\ldots,
x_{i_k}))$.

\begin{example}
  Let $\phi_r(x_1,x_2)$ be the formula from~\cref{ex:types} (where we
  renamed the variables to match the above definition of
  interpretations). Then $G_{\phi_r(x_1,x_2)}$ has the same vertex set as~$G$
  and any two vertices are joined by an edge in $G_{\phi_r(x_1,x_2)}$
  if, and only if, their distance is at most~$r$ in~$G$.
\end{example}

A formula
$\phi((x_{i_1},\ldots, x_{i_\ell}), (x_{i_{\ell+1}},\ldots, x_{i_k}))$
with an ordered partition of its free variables is \emph{stable on a class
  $\CCC$} of graphs, if there is an integer $s$ such that for every
graph $G\in \CCC$ the graph
$G_{\phi((x_{i_1},\ldots, x_{i_\ell}), (x_{i_{\ell+1}},\ldots,
  x_{i_k}))}$
has ladder index at most $s$.  A class $\CCC$ of graphs is
\emph{stable} if every formula with every partition of its free
variables is stable on $\CCC$.

As proved by Adler and Adler in \cite{adler2014interpreting} (based on work 
of Podewski and Ziegler~\cite{podewski1978stable}), stable
classes properly extend the concept of nowhere dense classes.

\begin{theorem}[\cite{adler2014interpreting}]\label{thm:ndstable}
  Every nowhere dense class of graphs $\CCC$ is stable, that is, for
  every formula
  $\phi((x_{i_1},\ldots, x_{i_\ell}), (x_{i_{\ell+1}},\ldots,
  x_{i_k}))$
  with every ordered partition of its free variables there is an integer $s$
  such that for every graph $G\in \CCC$ the graph
  $G_{\phi((x_{i_1},\ldots, x_{i_\ell}), (x_{i_{\ell+1}},\ldots,
    x_{i_k}))}$ has ladder index at most $s$.
\end{theorem}

Note that the converse is not true, as, e.g., the class of complete
graphs is stable but not nowhere dense. In particular, stable classes
are possibly not closed under taking subgraphs. We remark that
in~\cite{malliaris2014regularity} a stable class of graphs is defined as
a class $\CCC$ such that the ladder index of every graph $G\in \CCC$
is bounded by a constant $s$ depending only on $\CCC$ (that is,
\cite{malliaris2014regularity} does not demand closure under
interpretations).

The following theorem is easily proved using the Sauer-Shelah Lemma.

\begin{theorem}[see \cite{shelah1990classification}, Theorem II.4.10(4) and 
II.4.11(4), p.74]\label{thm:shelah1}
Let $\CCC$ be a stable class of graphs and let $\Delta$ be a finite
set of first-order formulas. There exists an integer $s$ such that for
all $G\in \CCC$ and all $A\subseteq V(G)$ with $\abs{A}\geq 2$ it holds
that $\abs{S_\Delta(G,A)}\leq \abs{A}^s$.
\end{theorem}

\begin{corollary}\label{crl:neighbourhoodcomplexity}
  Let $\CCC$ be a nowhere dense class of graphs and let $r$ be a
  positive integer. There is an integer~$s$ such that for all
  $G\in \CCC$ and all $A\subseteq V(G)$ it holds that
  \[\abs{\{N_r(v)\cap A : v\in V(G)\}}\leq \abs{A}^s.\]
\end{corollary}

\bigskip
\paragraph{Indiscernible sequences.}
Let $G$ be a graph and let $\Delta$ be a set of formulas. A sequence
$(v_1,\ldots, v_\ell)$ of vertices of $G$ is
\emph{$\Delta$-indiscernible} if for every formula
$\phi(x_1,\ldots, x_k)\in \Delta$ with $k$ free variables and any two
increasing sequences
$1\leq i_1<\ldots <i_k\leq \ell, 1\leq j_1< \ldots< j_k\leq \ell$ of
integers, it holds that
\[G\models\phi(v_{i_1},\ldots, v_{i_k})\Leftrightarrow G\models\phi(v_{j_1},
\ldots, v_{j_k}).\]

The following theorem forms the basis of our construction. The proof
follows immediately from Theorem 3.5, Item (2) of
\cite{malliaris2014regularity} and parallels that proof, 
however, we provide a proof of the theorem because we will 
provide a precise analysis for the nowhere dense case in \cref{sec:uqw}.

\begin{theorem}\label{thm:extract_indiscernibles}
  Let $\CCC$ be a stable class of graphs and let $\Delta$ be a finite
  set of first-order formulas.  There is a polynomial $p(x)$ such that
  for all $G\in \CCC$, every positive integer $m$ and every sequence
  $(v_1,\ldots, v_\ell)$ of vertices of $G$ of length $\ell=p(m)$, there
  exists a sub-sequence $(v_{i_1},\ldots, v_{i_m})$ of
  $(v_1,\ldots, v_\ell)$ of length~$m$ which is
  $\Delta$-indiscernible, $1\leq i_1<\ldots <i_m\leq \ell$.

  Furthermore, there is an algorithm that given an $n$-vertex graph $G\in \CCC$
  and a sequence $(v_1,\ldots, v_\ell)\subseteq V(G)$, computes a
  $\Delta$-indiscernible sub-sequence of $(v_1,\ldots, v_\ell)$ of length at
  least $m$. The running time of the algorithm is in
  $\Oof(\abs{\Delta}\cdot k \cdot \ell^{k+1} \cdot n^{q}\cdot a(n)\cdot \lambda(\Delta))$, where $k$
  is the maximal number of free variable, $q$ is the maximal quantifier-rank of
  a formula of $\Delta$, $a(n)$ is the time required to test adjacency between two vertices and
  $\lambda(\Delta)$ is the length of a longest formula in $\Delta$.    
\end{theorem}
\begin{proof}
  Let $G\in \CCC$ and let $(v_1,\ldots, v_\ell)$ be a sequence of
  vertices of $G$. We prove that for every formula
  $\phi(x_1,\ldots, x_k)\in \Delta$ there is a
  $\{\phi\}$-indiscernible subsequence of length at least $f^k(\ell)$,
  where $f(\ell)=\left(\frac{\ell}{t}\right)^{\frac{1}{tr+t+1}}-k$ for
  constants $r,t$ depending only on $\phi$ and $\CCC$. The claim of
  the theorem then follows by iteratively extracting
  $\{\phi\}$-indiscernible sequences for all $\phi\in \Delta$ and
  combining the polynomials accordingly.

  \medskip Let $\phi(x_1,\ldots, x_k)\in \Delta$. Let $G\in \CCC$ and let
  $(v_1,\ldots, v_\ell)$ be a sequence of vertices of~$G$. Let
  $A\coloneqq\{v_1,\ldots, v_\ell\}$. As $\CCC$ is stable, according
  to \cref{thm:shelah1}, there is an integer $r$ such that
  $\abs{S_{\{\phi\}}(G,A)}\leq \abs{A}^r$. Here, we consider the sets
  $S_{\{\phi((x_i), (x_1,\ldots, x_{i-1},x_{i+1},\ldots, x_k))\}}(G,A)$ for
  every partition $\phi((x_i), (x_1,\ldots, x_{i-1}$, $x_{i+1},\ldots, x_k))$ of
  the free variables of $\phi$ and we choose $r$ such that it works for every
  such partition of the variables.

  \medskip We prove by induction on $0\leq m\leq k$ that there exists
  a sub-sequence $\bar{u}_m$ of $(v_1,\ldots, v_\ell)$
  with
    \begin{enumerate}
    \item $\abs{\bar{u}_m}\ge f^m(\ell)$ and
    \item \label{item:cond-b}for all sub-sequences
      $\bar a = a_1,\ldots,a_k$ and $\bar b = b_1,\ldots,b_k$ of
      $\bar{u}_m$, if $a_i=b_i$ for all $i\in\{1,\ldots,k-m\}$, then
      \[G\models \phi(\bar a) \Leftrightarrow G\models \phi(\bar b)\,.\]
    \end{enumerate}
    Note that the elements of a sub-sequence of a sequence preserve their respective order.

    \bigskip For $m=0$ there is nothing to show, we can take
    $\bar{u}_0 =(v_1,\ldots, v_\ell)$ of length $f^0(\ell)=\ell$.

    For $m+1$ assume that $\bar{u}_m=(v_1,\ldots, v_{\ell(m)})$ is
    constructed as required. We define
    \[\phi^m = \phi((x_{k-m}), (x_1,\ldots, x_{k-m-1}, v_{\ell(m)-m+1},\ldots,v_{\ell(m)})),\]
    that is, we fix the interpretation of the last $m$ free variables
    as the last $m$ elements of the sequence~$\bar{u}_m$.

    \medskip We construct a tree $T$ whose vertices are elements from
    $\bar{u}_m$, except for the root, which we label~$\emptyset$. 
    We attach $v_1$ as a child of the root $\emptyset$. 
    In the following, by a \emph{maximal} element $z$ in the tree $T$
    satisfying some condition we always mean an element of $T$
    satisfying this condition which is as far away from the root as
    possible in the sense that in the subtree of $T$ rooted at $z$ no other
    element satisfies this condition.

    By induction,
    assume that for some $j < \ell(m)$ all $v_i$ with $i<j$ are
    inserted in the tree and we want to insert $v_j$. For a vertex
    $v$ of the tree let $P(v)$ be the path from the root to $v$
    excluding~$v$ and for a child $v$ of the root
    let $P(v)$ be the empty path. Let
    $w$ be a maximal vertex that realises the
    same $\{\phi^m\}$-type over $P(w)$ as $v_j$ (that is, no successor
    $z$ of $w$ realises the same $\{\phi^m\}$-type over~$P(z)$ as $v_j$), or
    $w = \emptyset$ if no such vertex exists. We insert $v_j$ as a new child
    of $w$.

    We call the resulting tree~$T$ a \emph{type tree.} For
    two vertices $v$ and $w$ we write $v <_T w$ if $v \in P(w)$.  The
    tree constructed in this way satisfies three properties:
    \begin{itemize}
    \item if $v = v_i <_T w = v_j$, then $i<j$, that is, paths in the type tree respect 
    the order of $\bar{u}_m$,
    \item if $v <_T w$, then $v$ and $w$ have
      the same $\{\phi^m\}$-types over $P(v)$, and
    \item if $v\neq w$ and neither $v <_T w$ nor $w <_T v$, then $v$
      and $w$ realise different $\{\phi^m\}$-types over
      $P(v)\cap P(w)$ and they realise the same $\{\phi^m\}$-type
      over $P(z)$, where $z$ is the maximal element of $P(v)\cap P(w)$.
    \end{itemize}

    We show that there exists a sufficiently long branch of $T$ that can be used as
    $\bar{u}_{m+1}$. 
    Let $\bar a = a_1,\ldots,a_k$ and
    $\bar b = b_1,\ldots,b_k$ be sub-sequences of $\bar{u}_m$ such
    that, first, $a_i=b_i$ for all $i\in\{1,\ldots,k-m-1\}$, and,
    second, both $\bar{a}$ and $\bar{b}$ are subsequences of a path
    $P(w)$ in $T$ for some~$w$. If $a_{k-m} = b_{k-m}$,
    Condition~\ref{item:cond-b} is proven for $\bar{u}_{m+1}$, so we can
    assume without loss of
    generality that $a_{k-m}<_T b_{k-m}$. For all sequences $\bar d_1$
    and $\bar d_2$ we abbreviate
    $G \models \phi(\bar d_1)\Leftrightarrow G\models \phi(\bar d_2)$
    by $\phi(\bar d_1)\equiv_G \phi(\bar d_2)$.  Then
    \[\phi(a_1,\ldots,a_{k-m-1},a_{k-m}, \ldots,a_k) \equiv_G
      \phi(a_1,\ldots,a_{k-m-1},a_{k-m}, v_{\ell(m)-m+1},\ldots,v_{\ell(m)})\]
    by the induction hypothesis, as the arguments coincide on the
    first $k-m$ arguments. As 
    $a_{k-m}$ and $b_{k-m}$ realise the same $\phi^m$-type over
    $P(a_{k-m})$, we obtain
    \[ \phi(a_1,\ldots,a_{k-m-1},a_{k-m},v_{\ell(m)-m+1},\ldots,v_{\ell(m)})
      \equiv_G
       \phi(a_1,\ldots,a_{k-m-1},b_{k-m},v_{\ell(m)-m+1},\ldots,v_{\ell(m)})
    \]
    and, using the induction hypothesis again and the condition that
    $a_i=b_i$ for $i\le k-m-1$,
    \[
      \phi(a_1,\ldots,a_{k-m-1},b_{k-m},v_{\ell(m)-m+1},\ldots,v_{\ell(m)})
      \equiv_G
      \phi(a_1,\ldots,a_{k-m-1},b_{k-m},b_{k-m+1},\ldots,b_k)\,.
    \]
    This proves Condition~\ref{item:cond-b}. 
    
    It remains to show that the tree
    has a sufficiently long branch. For this, we want to show that $T$ does not
    contain a large complete binary subtree. 
	Consider the graph
    \[G_B=G_{\phi((x_{k-m}),(x_1,\ldots, x_{k-m-1},v_{\ell(m)-m+1},\ldots,
      v_{\ell(m)}))}\]
    consisting of vertices represented by single vertices of~$G$ (first type)
    and those represented by $k-1$-tuples (second type). 
	As by assumption $\CCC$ is stable,
    according to \cref{thm:ndstable}, the graph $G_B$ as
    an interpretation of $G$ has
    ladder index at most $s$ for some constant~$s$, and hence according
    to \cref{lem:branching}, its branching index is bounded by
    $t\coloneqq 2^{s+2}-2$. Let $S$ be a complete binary subtree of $T$. 
	We show that we can construct a branching witness for $G_B$ with 
	the same leaves as $S$ (note that the vertices of first type are actual 
	vertices of $G$). This implies that
    $S$ can have depth at most $t$. 
    
    By construction of $S$, all distinct leaves $v,w$ of $S$ have a 
    different $\{\phi^m\}$-type over $P(v)\cap P(w)$ while they
    have the same type over $P(z)$, where $z$ is the maximal element
    of $P(v)\cap P(w)$. In particular, there is a tuple 
    $(a_1,\ldots, a_{k-m-1})\in P(z)$ such that 
    $G\models\phi(a_1,\ldots, a_{k-m-1}, u, v_{\ell(m)-m+1},\ldots, 
    v_{\ell(m)})$ for all right successors $v$ of $z$, and such that 
    $G\not\models\phi(a_1,\ldots, a_{k-m-1}, w, v_{\ell(m)-m+1},\ldots, 
    v_{\ell(m)})$ for all left successors $w$ of $z$. 
    Observe that all inner vertices of the so constructed branching
    witness are distinct as no other tuple contains the element $a_{k-m-1}$
    at position $k-m$. Hence we constructed a branching witness of index $t$
    as claimed. 
    

    We now assign to every vertex $v$ in the tree $T$ its \emph{binary rank},
    that is, the maximal height of a full binary tree that is a
    sub-graph of the sub-tree rooted at $v$ (compare to the
    \emph{stability rank} of sets of formulas in~\cite{shelah1990classification}). The
    \emph{depth} of a vertex $v$ in $T$ is $\abs{P(v)}$.
  
    Let $N^s_\ell$ be the set of vertices of the type tree $T$ with
    rank $s$ and depth $\ell$. For $\ell>0$ let
    $X^s_\ell\subseteq N^s_\ell$ be the set vertices from $N^s_\ell$
    whose direct predecessor has rank $s$ and let
    $Y^s_\ell\subseteq N^s_\ell$ be the set vertices from~$N^s_\ell$
    whose direct predecessors have rank $s+1$ (we may assume that we have maximal branching, 
    hence all predecessors have rank $s$ or $s+1$).
  
    Then $N^s_\ell \subseteq X^s_\ell \cup Y^s_\ell$.
    Define
    \[n^s_\ell = \abs{N^s_\ell}, \qquad x^s_\ell = \abs{X^s_\ell}
    \text{\qquad and } \qquad y^s_\ell = \abs{Y^s_\ell}\,.\]
  
    A vertex $v$ of depth $d$ can have at most $(d+m+1)^r$ direct
    successors (recall that for every set $A$ we have
    $\abs{S_{\{\phi\}}(G,A)}\leq \abs{A}^r$). This is because every two such
    successors have different $\phi^m$-types over $P(v)\cup\{v\}$, so
    there are $\abs{P(v)}+1$ predecessors and, additionally, $m$
    parameters that are fixed in the formula~$\phi^m$.

    The following inequalities hold:
    \begin{enumerate}
    \setlength\itemsep{3pt}
    \item $x^s_{\ell+1} \le n^s_\ell$, because the direct predecessor $w$ of any vertex in
      $X_\ell^s$ has at most one direct successor
      of rank $s$ (otherwise $w$ had rank $s+1$),\label{item:a}
    \item
      $y^s_{\ell+1} \le n^{s+1}_\ell\cdot (\ell+m+1)^r$,\label{item:b}
    \item
      $n^s_{\ell+1} \le n^s_\ell + n^{s+1}_\ell\cdot
      (\ell+m+1)^r$,\label{item:c}
    \item $n^s_0 = 0$ for $1\le s<t$,\label{item:d}
    \item $n^t_\ell \le 1$ where $t$ is the rank of the
      root. \label{item:e}
    \end{enumerate}

\smallskip 
    \begin{claim}
    $n^{t-s}_{\ell+1} \le (\ell+m+1)^{s(r+1)}$.
    \end{claim}
    \smallskip
    \begin{claimproof}
      For $s=1$ we show the statement by induction on $\ell$. For
      $\ell=0$ first (\ref{item:c}) and then (\ref{item:d}) and
      (\ref{item:e}) give us
      \[n^{t-1}_1 \le n^{t-1}_0 + n^t_0(\ell+m+1)^r \le
      (\ell+m+1)^r\,.\]
      For $\ell+1$ using first (\ref{item:c}) and then the induction
      hypothesis (for $\ell$) and (\ref{item:e}) we obtain
      \[
      n^{t-1}_{\ell+1} \le n^{t-1}_\ell + n^t_\ell (\ell+m+1)^r \le
      (\ell + m)^{r+1} + (\ell+m + 1)^r \le (\ell + m +1)^{r+1}\,.
      \]

      For $s+1$ the proof is again by induction on $\ell$.  For
      $\ell=0$,
      \[\ n^{t-(s+1)}_1 \le n^{t-(s+1)}_0 + n^{t-s}_0(\ell + m + 1)^r
      \le (\ell+m+1)^r\]
      and for $s+1$ as above and using the induction hypothesis for
      $s$,
      \begin{align*}
      n^{t-(s+1)}_{\ell+1} \le n^{t-(s+1)}_\ell + n^{t-s}_\ell(\ell
      + m + 1)^r& \le (\ell+m)^{(s+1)(r+1)} + (\ell+m)^{s(r+1)}(\ell+m+1)^r \\&\le
      (\ell+m+1)^{(s+1)(r+1)}\,.
      \end{align*}
    \end{claimproof}

    The total number $n_{\ell+1}$ of vertices of depth $\ell+1$ is
    then
    \[ n_{\ell+1} \le \sum_{s\le t} n^{t-s}_{\ell+1} \le \sum_{s\le t}
    (\ell+m+1)^{s(r+1)} \le t(\ell+m+1)^{t(r+1)}\,.\]
    Now the number of vertices in a tree of depth at most $d$
    (including the root at depth~$0$) is
    \[ N = 1 + \sum_{\ell < d} n_{\ell+1} \le 1 +
    \sum_{\ell<d}t(\ell+m+1)^{t(r+1)} < t(d+m+1)^{t(r+1)+1},\]
    and thus $d> \left(\frac N{t}\right)^{\frac{1}{tr+t+1}}-m-1$, so if
    $\abs{\bar{u}_m}>t(d+m+1)^{t(r+1)+1}$, then there is a branch of
    length at least $\left(\frac N{t}\right)^{tr+t+1}-m$. Replacing
    $m$ with $k$ we obtain a slightly worse bound that, however, does
    not depend on the induction step.

    In $k$ steps we extract a sequence $\bar{u}$ of length at least
    $f^{(k)}(\ell(m))$, where
    $f(\ell) = \left(\frac \ell{t}\right)^{\frac{1}{tr+t+1}}-k$.

    \bigskip It remains to analyse the running time of the
    constructive procedure described above.  It suffices to
    show that the sequence constructed in the inductive step of the above
    proof can be
    computed in polynomial time. First, the type tree $T$ is
    computed. We construct $T$ inductively as in the proof. While
    adding new vertices to the tree, we keep track of their height in
    the tree and we keep track of the longest branch. Thus after
    computing~$T$, we just output the longest branch as $\bar{u}_{m+1}$.

    For every vertex $v\in \bar{u}_i$, we
    search through the type tree to find the maximal element with the
    same $\{\phi^m\}$-type to decide where to insert $v$. Here, we have
    to compare $v$ to less than $\ell$ elements. 

    To check the $\{\phi^m\}$-type we perform a standard model-checking
    algorithm to verify whether $G\models \phi(\bar a, v, \bar
    v)$.
    Here $\bar a$ is an $m$-tuple of vertices on the path from the root to the
    current leaf (of length less than $\ell$) 
    and~$\bar v$ are the parameters from $\phi^m$. 
	Note that $m\le k$ and the parameters in the formulas of
    $\phi^m$ are fixed for each $m$. Hence, the check whether 
    $G\models \phi(\bar a, v, \bar v)$ can be carried out in time
    $\Oof(n^q\cdot \lambda(\Delta)\cdot  a(n))$ and it takes time at most
    $\Oof(\ell^k)$ to iterate through all
    $m$-tuples~$\bar a$.

    Summing up, we need $\Oof(\abs{\Delta}\cdot k \cdot\ell\cdot \ell^k
    \cdot n^{q}\cdot \lambda(\Delta)\cdot a(n))=\Oof(\abs{\Delta}\cdot k \cdot \ell^{k+1}
    \cdot n^{q}\cdot \lambda(\Delta)\cdot a(n))$ steps to
    compute the final $\Delta$-indiscernible sequence. 
\end{proof}

As a consequence we obtain the following corollary, which 
in particular applies to nowhere dense classes of graphs.

\begin{corollary}\label{crl:extract_indiscernible_nd}
  Let $\CCC$ be a stable class of graphs and let~$\Delta$ be an arbitrary finite set of
  formulas. Then there is a polynomial $p(x)$ such that for every
  positive integer~$m$ and every sequence $(v_1,\ldots, v_\ell)$ of
  vertices of $G\in \CCC$ of length $\ell=p(m)$, there is a
  subsequence of length at least $m$ which is $\Delta$-indiscernible.
\end{corollary}

\section{Improved Bounds on Uniform Quasi-Wideness}
\label{sec:uqw}

We can now show our first main theorem (Theorem~\ref{thm:uqw}), stating that a class is
uniformly quasi-wide with margin $N\colon\N\times \N\rightarrow \N$ if, and
only if, for every $r\in \N$ there is a polynomial margin
$N'(r,m)$. For convenience, we repeat the statement of the theorem
here. 

\newcounter{tmp_c}
\setcounter{tmp_c}{\value{theorem}}
\setcounter{theorem}{\value{main_thm}}
\newcounter{tmp_c_section}
\setcounter{tmp_c_section}{\value{section}}
\setcounter{section}{\value{main_thm_section}}

\begin{theorem}
Let $\CCC$ be a nowhere dense class of graphs. For every $r\in \N$
  there exists a polynomial~$p_r(x)$ and a constant $s(r)$ such that
  for all $m\in \N$ the following holds. For all $G\in\CCC$ and all 
  sets $A\subseteq V(G)$ of size at least
  $p_r(m)$, there is a set $S\subseteq V(G)$ of size at
  most $s(r)$ such that there is a set $B\subseteq A$ of size at
  least~$m$ which is $r$-independent in $G-S$.

  Furthermore, if $K_c\not\minor_rG$ for all $G\in \CCC$, then
  $s(r)\leq c$ and there is an algorithm, that given an $n$-vertex graph
  $G\in \CCC$, $\epsilon>0$, $r\in \N$ and $A\subseteq V(G)$ of size at least
  $p_r(m)$, computes a set $S$ of size at most $s(r)$ and an
  $r$-independent set $B\subseteq A$ in $G-S$ of size at least $m$ in
  time $\Oof(r\cdot c\cdot |A|^{c+1}\cdot n^{1+\epsilon})$.
\end{theorem}
\setcounter{theorem}{\value{tmp_c}}
\setcounter{section}{\value{tmp_c_section}}

The proof of the theorem is based on the extraction of a large
$\Delta$-indiscernible sequence from the set $A$ for a certain set
$\Delta$ of formulas.  Let $k\in \N$ and for $1\leq i\leq k$ let
\smallskip
\[\phi_i(x_1,\ldots, x_k)\coloneqq\exists y\big(\bigwedge_{1\leq j\leq
  i}E(y,x_j) \wedge \bigwedge_{i<j\leq k}\neg E(y,x_j)\big)\] and let
\[\psi_i(x_1,\ldots, x_k)\coloneqq\exists y\big(\bigwedge_{i< j\leq
  k}E(y,x_i) \wedge \bigwedge_{1\leq j\leq i}\neg E(y,x_j)\big).\]

The formula $\phi_i(x_1,\ldots,x_k)$ states that there exists a 
vertex $v$ which is adjacent exactly to the first~$i$ elements 
$x_1,\ldots, x_i$, while $\psi_i(x_1,\ldots,x_k)$ states that 
there exists a vertex which is non-adjacent exactly to the
first $i$ elements $x_1,\ldots, x_i$. 
Let
\[\Delta_k\coloneqq\{E(x_1,x_2)\}\cup \{\phi_1,\ldots, \phi_k, \psi_1,\ldots,\psi_k\}.\]

\bigskip Note that the formula $E(x_1,x_2)$ in $\Delta_k$ guarantees
that the vertices of every $\Delta_k$-indiscernible sequence of
vertices of a graph $G$ either form an independent set or a clique in
$G$.

The crucial property we are going to use is stated as Claim 3.2 in
\cite{malliaris2014regularity}. Recall the definition of the
ladder-index from \cref{sec:prelims}.

\begin{lemma}[Claim 3.2 of \cite{malliaris2014regularity}]\label{lem:smalldegree}
  Let $k\in \N$ and let $G$ be a graph with ladder-index less
  than~$k$. If $n\geq 4k$ and $(v_1,\ldots, v_n)$ is a
  $\Delta_k$-indiscernible sequence in $G$ and $w\in V(G)$, then
  either
  \[\abs{N(w)\cap \{v_1,\ldots, v_n\}}<2k\text{ or }\abs{\{v_1,\ldots,
  v_n\}\setminus N(w)}<2k.\]
\end{lemma}

For nowhere dense classes we can derive even stronger properties of
$\Delta_k$-indiscernible sequences. We need one more lemma, which follows easily
from Lemma 4.11(2) of \cite{gajarsky2017kernelization}.

\begin{lemma}[see Lemma 4.11(2) of \cite{gajarsky2017kernelization}]\label{lem:diversity}
  Let $\CCC$ be a nowhere dense class of graphs. For every
  $\epsilon>0$ there is an integer $n_0$ such that if
  $A\subseteq V(G)$ for $G\in \CCC$ with $\abs{A}\geq n_0$,
  then \[\abs{\{N(v)\cap A : v\in V(G)\}}\leq \abs{A}^{1+\epsilon}.\]
\end{lemma}

Note that if $K_k\not\minor_1G$, then $G$ does not contain a
$k$-ladder and excludes the complete bipartite graph
$K_{k,k}$ as a subgraph.

\begin{lemma}\label{lem:realisedtypes}
  Let $G$ be a graph with $K_k\not\minor_1 G$. There exists an integer
  $n_1=n_1(k)$ such that if $(v_1,\ldots v_n)$ is a
  $\Delta_k$-indiscernible sequence of length $n\geq n_1$, then every
  vertex $v\in V(G)$ is either connected to at most one vertex of
  $\{v_1,\ldots, v_n\}$ or to all of them.
\end{lemma}
\begin{proof}
  As $K_k\not\minor_1G$, in particular, $G$ has ladder-index less than
  $k$ and excludes $K_{k,k}$ as a subgraph. If $n\geq 4k$,
  according to \cref{lem:smalldegree}, every vertex $w\in V(G)$
  satisfies $\abs{N(w)\cap \{v_1,\ldots, v_n\}}<2k$ or
  $\abs{\{v_1,\ldots, v_n\}\setminus N(w)}<2k$. We first show that there
  are only a few vertices $w$ which satisfy
  $\abs{\{v_1,\ldots, v_n\}\setminus N(w)}<2k$. We will refer to these
  vertices as high degree vertices in the rest of the proof.

  Assume there are at least $k$ high degree vertices. Fix a set $A$ of
  exactly $k$ of these vertices. Take
  as $B$ the set $(\bigcap_{w\in A}N(w))\cap \{v_1,\ldots, v_n\}$.
  This set has order at least $n-2k^2$. By definition of
  $\Delta_k$-indiscernibility, the vertices $v_1,\ldots, v_n$ either
  form an independent set or a clique in $G$. If $n\geq k$, they form
  an independent set by assumption. Hence
  $A\cap B= \emptyset$ and if $n\geq 2k^2+k$, we
  find a complete bipartite graph $K_{k,k}$ as a subgraph of $G$, contradicting our
  assumption.

  \smallskip Assume towards a contradiction that there is a vertex $v\in V(G)$ which is connected
  to exactly $s$ vertices among $\{v_1,\ldots, v_n\}$,
  $2\leq s\leq n-1$.  Then for some $i$, $2\leq i\leq k-1$, there is
  an increasing tuple $(v_{i_1},\ldots, v_{i_k})$ for
  $1\leq i_1<\ldots <i_k\leq n$ such that
  $G\models\phi_i( v_{i_1},\ldots, v_{i_k})$ or
  $G\models\psi_i(v_{i_1},\ldots, v_{i_k})$: if $s\leq 2k$, we can
  choose $i=2$, pick $2$ neighbours of $v$ among $v_1,\ldots,v_n$ and $k-2$ non-neighbours
  which are either all smaller or all larger than the two neighbours
  and define $(v_{i_1},\ldots, v_{i_k})$ accordingly. Similarly, if
  $s>n-2k$ we can choose $i=k-1$ and pick one non-neighbour of $v$ and
  $k-1$ neighbours to define $(v_{i_1},\ldots, v_{i_k})$.
  
  Consider first the case that there is a tuple $(v_{i_1},\ldots, v_{i_k})$ such
  that $G\models\phi_2(v_{i_1},\ldots, v_{i_k})$ (the case that
  $G\models \psi_2(v_{i_1},\ldots, v_{i_k})$ is completely analogous).  By
  definition of $\Delta_k$-indiscernibility, every increasing
  $k$-tuple $(v_{j_1},\ldots, v_{j_k})$ satisfies $\phi_2$, that is,
  there is an element $w\in V(G)$ such that $w$ is adjacent to
  $v_{j_1}$ and $v_{j_2}$ and not to $v_{j_3},\ldots, v_{j_k}$.  There
  are ${n \choose k}$ increasing $k$-tuples of elements of
  $(v_1,\ldots, v_n)$.
  
  Every vertex $z$ with $a \coloneqq \abs{\{v_1,\ldots, v_n\}\setminus N(z)}<2k$
  can take the role of this vertex $w$ for at most
  ${n-a \choose 2}\cdot {a\choose k-2}$ tuples: choose $2$ vertices
  from the $n-a$ neighbours of $w$ and $k-2$ vertices from its $a$
  non-neighbours among $\{v_1,\ldots, v_n\}\setminus N(z)$.  We have
  at most $k$ of these high degree vertices as shown above. Hence at
  most $k\cdot {n-1 \choose 2}\cdot {2k-1\choose k-2}\in \Oof(n^2)$
  tuples satisfy $\phi_2$ because of high degree vertices.
  
  Every vertex $z$ with $a\coloneqq\abs{\{v_1,\ldots, v_n\}\cap N(z)}<2k$ can
  play the role of $w$ for at most
  ${a \choose 2}\cdot {n-a \choose k-2} \leq {2k-1\choose 2}{n-2
    \choose k-2}\in \Oof(n^{k-2})$
  tuples.  Denote by $x$ the number of small degree vertices. 
  
  \vspace{1mm}
  \noindent Then it
  must hold that
  \[k\cdot {n-1 \choose 2}\cdot {2k-1\choose k-2} + x\cdot
  {2k-1\choose 2}{n-2 \choose k-2}\geq {n\choose k},\]
  from which we conclude that $x\geq c\cdot n^2$ for some fixed
  constant $0<c<1$ and all $n>n_1$ for sufficiently large $n_1$ (choose
  $\epsilon=1/2$ and $n_1$ large enough such that we can 
  apply \cref{lem:diversity}).
  Without loss of generality we may assume that all small degree
  vertices induce distinct neighbourhoods on~$A$ (realising one
  neighbourhood twice does not help to realise more types). Now, for
  sufficiently large $n$ we have a contradiction
  to \cref{lem:diversity}.
  
  The proof for the case that there is a tuple
  $(v_{i_1},\ldots, v_{i_k})$ such that
  $G\models\phi_{k-1}(v_{i_1},\ldots, v_{i_k})$ or $G\models
  \psi_{k-1}(v_{i_1},\ldots, v_{i_k})$ is similar. Here, the
  high degree vertices can cover at most
  $k\cdot {n-1 \choose k-1}\cdot (2k-1)\in \Oof(n^{k-1})$ many tuples,
  while every low degree vertex can cover at most
  ${2k-1\choose k-1}\cdot (n-k-1)\in \Oof(n)$ many tuples.
\end{proof}

We are ready to prove the main theorem.

\begin{proof}[Proof of \cref{thm:uqw}]
  Let $\CCC$ be a nowhere dense class of graphs such that \mbox{$K_{f(i)}\not\minor_i
  G$} for all~$i$ and all $G\in \CCC$.  Let $G\in \CCC$ and let
  $A\subseteq V(G)$. According to \cref{crl:extract_indiscernible_nd}, we can
  extract from $A$ a large $\Delta_{f(1)}$-indiscernible sequence $B_1$.  This
  requires $A$ to be only polynomially larger than~$B_1$ according to the
  corollary.

  We construct a sequence of graphs $G_1 = G, G_2, \ldots$ and sequences
  $B_1\supseteq B_2\supseteq \ldots$, $\tilde{B}_1\supseteq \tilde{B}_2\supseteq \ldots$  and $S_1,S_2,\ldots$ of
  vertices of $G_i$ where $B_i \subseteq V(G)$, $\tilde{B}_i\subseteq V(G_i)$ and
  $S_i\subseteq V(G)\cap V(G_i)$ such that for all $i$
  \begin{enumerate}
  \item $B_i$ is $2i$-independent in $G-Z_i$ where
    $Z_i \coloneqq \bigcup_{j= 1}^{i} S_j$,\label{prop1}
  \item $\abs{S_i} < f(i)$, 
  \item $G_{i+1} \minor_i G$ and
  \item $\tilde{B}_i = \{G[N_i^{G-Z_i}(v)] \mid v\in B_i\}$ (recall
    that the vertices of $G_i$ as a minor of $G$ are subgraphs of $G$).
  \end{enumerate}

  As $K_{f(1)}\not\minor_1G_1$, we can apply \cref{lem:realisedtypes} and
  conclude that if $B_1$ is sufficiently large, then every vertex
  $v\in V(G_1)$ is either connected to at most one vertex of
  $B_1$ or to all vertices of~$B_1$. Just as in the proof
  of \cref{lem:realisedtypes}, we show that there are less than $f(1)$
  vertices that are adjacent to all vertices of $B_1$. We define $S_1$ as the
  set of all those (less than $f(1)$) high degree vertices: $S_1 = \{v\in V(G)
  \mid N(B_1)\subseteq N(v)\}$. 

  Note that $B_1\cap S_1=\emptyset$ if $\abs{B_1}\geq f(1)$, as in this case, as
  above, the vertices of $B_1$ form an independent set. Hence every
  vertex $v\in V(G_1-S_1)$ is connected to at most one vertex of $B_1$, in other
  words, $B_1$ is a $2$-independent set in $G_1-S_1$. We hence
  established all of the above properties for $i=1$.

  Let the sequences $G_1, G_2, \ldots, G_i$, $B_1,B_2,\ldots,B_i$ and
  $S_1,S_2,\ldots,S_i$ be defined for some fixed~$i$. By
  Assumption \ref{prop1}, $B_i$ is $2i$-independent in
  $G-Z_i$, hence we can contract the disjoint
  $i$-neighbourhoods of the vertices of $B_i$. Let
  $\tilde{B}_i = \{G[N_i^{G-Z_i}(v)] \mid v\in B_i\}$ and let $H$ be the graph resulting from
  the contraction (\ie $\tilde{B}_i\subseteq V(H)$). Let $\tilde{B}_i' \subseteq \tilde{B}_i$ be a
  large independent set in~$H$. We can obtain it by finding a large
  $\{E(x_1,x_2)\}$-indiscernible subsequence of $\tilde{B}_i$ in $H$ (by
  \cref{thm:extract_indiscernibles}). Let
  $B_i'\subseteq B_i$ be the set of vertices of $G$ such that
  $\tilde{B}_i' = \{G[N_i^{G-Z_i}(v)] \mid v\in B_i'\}$, \ie the vertices of $B_i'$ are the
  centres of $\tilde{B}_i'$. Define the graph $G_i'$ to be the
  depth-$i$ minor of~$G$ obtained by contracting the disjoint $i$-neighbourhoods
  of vertices in $B_i'$. Let $C$ be a $\Delta_{f(i+1)}$-indiscernible
  subsequence of $\tilde{B}_i'$ in $G_i'$. Define~$S_{i+1}$ as in case $i=1$, \ie as
  the set of vertices $v$ with $C\subseteq N(v)$. Note that we constructed $G_i'$
  and did not work with $H$ because otherwise the vertices of $S_{i+1}$ could
  have been contracted vertices. Those vertices are subgraphs of $G$ and can
  be arbitrarily large. Hence we would possibly delete much more than
  $\abs{S_{i+1}}$ vertices.

Define $B_{i+1}\subseteq B_i'$ as the set
  of vertices with $C = \{G[N_i^{G-Z_i}(v)] \mid v\in
  B_{i+1}\}$. If $C$ is large enough (just as in the case $i=1$), $C$ is
  $2$-independent in $G_i'-Z_{i+1}$, so $B_{i+1}$ is
  $2(i+1)$-independent in $G-Z_{i+1}$.

  For $i=r$ we are left with a set of size $m$ if we started with a set of
  size $p_r(m)$ with $p_r(x)$ chosen appropriately, tracing back the
  construction. The vertices $Z_i$ we delete during the construction are connected
  with all (contracted) vertices $\tilde{B}_i$. Note that the vertices of $Z_r$ 
  we delete are connected to vertices at 
  distance at most $r-1$ to the vertices of $B_r$. 
  We can hence merge the vertices $z$ of~$Z_r$ with the first $|Z_r|$
  branch sets $\tilde{B}_{r}$ (each vertex $z$ with one distinct branch
  set) to build a complete depth-$r$ minor of order $|Z_r|$. 
  Hence we can take $s(r) \coloneqq f(r)$, as this bounds the size of
  a largest complete depth-$r$ minor by assumption.

  We now show how to compute the sets $B_i, \tilde{B}_i$ and $Z_i$. In iteration
  $i$ we want to compute a $\Delta_{f(i)}$-indiscernible sequence as in the
  proof of \cref{thm:extract_indiscernibles}.  We remark that usually a graph
  $G$ from a nowhere dense class will be stored as a list of adjacency
  lists. As, for every $\epsilon>0$, $G$ is $c$-degenerate for some
  $c\in \Oof(n^\epsilon)$, we can store the adjacency relation in a more
  efficient way. Let $L$ be a linear order such that for each $v\in V(G)$ we
  have $\abs{\{ w\in N(v) \mid w <_L v \}} \le c$, \ie every vertex has at most~$c$ smaller neighbours. Now for every vertex we store the set of its smaller
  neighbours. Then we can implement an adjacency test in time $\Oof(n^\epsilon)$
  for any fixed $\epsilon>0$. Model-checking for the formulas of $\Delta_i$ with
  one quantifier is hence possible in time $\Oof(n^{1+\epsilon})$ by
  \cref{thm:extract_indiscernibles}. Furthermore, we have to perform the
  first $r$ levels of breadth-first searches around the elements of $A$ to keep
  track of the sets $\bar{B}_i$. Note that the $i$-neighbourhoods of these
  elements are disjoint when the searches are taken one step further (after
  deleting the set $S_i$). Hence, every edge appears at most once in all
  searches, and the searches can be carried out in time
  $\Oof(n^{1+\epsilon})$. Here we use again that a sufficiently large graph from
  a nowhere dense class of graphs has at most $n^{1+\epsilon}$ edges for any
  fixed $\epsilon>0$.
\end{proof}

\section{A Polynomial Kernel for \textsc{Distance-$r$ Dominating Set}}
\label{sec:kernel}
We now show how to obtain a polynomial kernel for the \textsc{Distance-$r$
  Dominating Set} problem.  Recall that a \emph{kernelization
algorithm}, or short \emph{kernel}, for the
\textsc{Distance-$r$ Dominating Set} problem parameterized by the solution size
is a polynomial time algorithm, which on input $G$ and $k$ computes another
graph~$H$ and a new parameter $k'$ which is bounded by a function of $k$
(independent of $\abs{V(G)}$), such that~$G$ contains a distance-$r$ dominating set of size
at most $k$ if, and only if, $H$ contains a distance-$r$ dominating set of size at most
$k'$. By abuse of notation, we also call the output of a kernelization algorithm
a kernel. If $\abs{V(H)}$ is bounded by a polynomial in $k$, then the kernel is
called a \emph{polynomial kernel}. We remark that by our definition, a 
kernelization algorithm on a class $\CCC$ does not necessarily 
output graphs from $\CCC$. Very strictly speaking, we are hence computing
a bi-kernel and not a kernel (as we reduce not to an instance of 
the original problem "distance-$r$ dominating set on $\CCC$" but to
the formally different problem "distance-$r$ dominating set). 

The idea is to kernelise the instance in two phases. 
In the first phase we reduce the number of \emph{dominatees}, that is, 
the number of those vertices whose domination is essential. More precisely, 
for an integer $k$, 
a set $Z\subseteq V(G)$ is called an \emph{$r$-domination core for parameter $k$} 
if every set 
$X\subseteq V(G)$ of size at most $k$ which $r$-dominates $Z$ also
$r$-dominates $V(G)$. In the second phase we reduce the number of \emph{dominators}, 
that is, the number of vertices that shall be used to dominate other vertices. 

\pagebreak
For the first phase, we argue just as Dawar and Kreutzer in 
Lemma~11 of~\cite{DawarK09}, to obtain an $r$-domination core in $G$. The key
idea of the lemma is to find a large $2r$-independent set $A$ after deleting at most
$s$ vertices, such that at least two elements of $A$ behave alike with respect to
every small dominating set.

Fix a nowhere dense class $\CCC$ of graphs and let $N(m,r)=p_r(m)$ for the
polynomial $p_r(x)$ and $s(r)$ characterising $\CCC$ as uniformly quasi-wide
according to \cref{thm:uqw}. Fix positive integers~$r$ and $k$ and let
$s\coloneqq s(2r)$. Let $c$ be the minimum integer such that 
$K_c\not\minor_{2r}G$ for
all $G\in \CCC$.

The proof of the following lemma is the same as in~\cite{DawarK09}, we
just use better bounds from \cref{thm:uqw}.

\begin{lemma}[see \cite{DawarK09}, Lemma 11]
  For $r\geq 0$ let $p_r$ be the polynomial defined in~\cref{thm:uqw}.
  There is an algorithm
  that, given an
  $n$-vertex graph $G\in\CCC$, $\epsilon>0$, $k>0$ and $Z\subseteq V(G)$ with 
  $\abs{Z}>p_{2r}\bigl((k+2)(2r+1)^s\bigr)\eqqcolon \ell$ runs in time $\Oof(s\cdot \ell^{c+1}\cdot 
  r\cdot c\cdot n^{1+\epsilon})$, and  
  returns a vertex $w\in Z$ such that for any set $X\subseteq V(G)$
  with $\abs{X}\leq k$,
  \[X \text{ $r$-dominates $Z$ if, and only if, $X$ $r$-dominates
    $Z\setminus \{w\}$}.\]
\end{lemma}
\begin{proof}
Fix a set $A\subseteq Z$ of size exactly $\ell$. By \cref{thm:uqw}
we can find in time $\Oof(r\cdot c\cdot \ell^{c+1}\cdot n^{1+\epsilon})$
sets $B\subseteq A$ and $S\subseteq V(G)$ such that $|S|\leq s$ and 
$|B|\geq (k+2)(2r+1)^s$ such that $B$ is $2r$-independent in $G-S$. 
Let $S=\{t_1,\ldots, t_s\}$ and, for each $v\in B$, compute
the distance vector $d_v=(d_1,\ldots, d_s)$, where $d_i=\mathrm{dist}(v,t_i)$
if this distance is at most $2r$ and $d_i=\infty$, otherwise. For this, we
have to perform $s$ breadth-first searches which takes time 
$\Oof(s\cdot n^{1+\epsilon})$. Note that there are at most $(2r+1)^s$
distinct distance vectors. Since $\ell\geq (k+2)(2r+1)^s$, there are 
at least $k+2$ elements $b_1,\ldots, b_{k+2}\in B$ which have
the same distance vector. Now we choose $w\coloneqq b_1$ and show that
for any set $X \subseteq V(G)$ with
$|X| \leq k$, $X$ $r$-dominates $Z$ if, and only if, $X$ $r$-dominates $Z
\setminus \{b_1\}$.

The direction from left to right is obvious.  Now, suppose $X$ $r$-dominates
$Z \setminus \{b_1\}$.  Consider the sets $B_i \coloneqq N_r^{G-S}(b_i)$ for $i \in
[2,\ldots,k+2]$.  These sets are, by construction, mutually disjoint.  Since there
are $k+1$ of these sets, at least one of them, say $B_j$, does not contain any
element of $X$.  However, since $b_j \in Z\setminus \{b_1\}$ there is a path
of length at most $r$ from some element $x$ in $X$ to $b_j$.  This path must,
therefore, go through an element of $S$.  Since $b_1$ and $b_j$ have
the same distance vector, we conclude that there is also a path of length at most $r$ from $x$ to $b_1$
and therefore $X$ $d$-dominates $Z$.

For the complexity bounds, note that all distance vectors can be computed
in time $\Oof(|S|\cdot |A| \cdot |G|) = \Oof(s\cdot p_{2r}\bigl((k+2)(2r+1)^s\bigr) \cdot n)$
(recall that $G$ is degenerate).  Adding this to the
$\Oof(r\cdot c\cdot \ell^{c+1}\cdot n^{1+\epsilon})$ time to find $S$ and $A$ gives us the
required bound.
\end{proof}

We now proceed as follows. We let $Z=V(G)$ and apply the above procedure to 
remove an irrelevant element from the $r$-domination core $Z$ until $\abs{Z}\leq p_{2r}\bigl((k+2)(2r+1)^s\bigr)$.  

\begin{corollary}\label{crl:domcore}
Let $\CCC$ be a nowhere dense class of graphs and let $k,r\in \N$. There is 
an algorithm running in time $\Oof(n^{2+\epsilon})$ that given an 
$n$-vertex graph $G\in\CCC$ and $\epsilon>0$, computes an 
$r$-domination core for parameter
$k$ of $G$ of size polynomial in~$k$. 
\end{corollary}

We now reduce the number of \emph{dominators}, 
that is, the number of vertices that shall be used to dominate other vertices. 
For this, observe that only vertices at distance at most $r$ to a vertex from 
the $r$-domination core are relevant. Furthermore, if there are two vertices
$u,v\in V(G)$ with $N_r(u)\cap Z=N_r(v)\cap Z$, it suffices to keep one 
of $u$ and $v$ as a representative. 

\begin{lemma}\label{lem:representatives}
Let $\CCC$ be a nowhere dense class of graphs and let $r\in \N$. 
There is an algorithm running in time $\Oof(n^{1+\epsilon}\cdot |Z|^{1+\epsilon})$, that given an 
$n$-vertex graph $G\in \CCC$, $\epsilon>0$ and $Z\subseteq V(G)$ computes a minimum size 
set $Y\subseteq V(G)$ such that for all $u\in V(G)$ there is $v\in Y$ with 
$N_r(u)\cap Z = N_r(v)\cap Z$. 
\end{lemma}
\begin{proof}
First, each element $v\in V(G)$ learns the set $N_r(v)\cap Z$. For this, 
we perform $|Z|$ breadth-first searches of depth $r$ starting
at the elements of $Z$. Whenever a breadth-first search around
$z\in Z$ encounters the element $v$, it registers the element $z$
in a list of $v$. For any $\epsilon>0$, this takes time 
$\Oof(|Z|\cdot n^{1+\epsilon})$ on a graph from a nowhere dense class. 

Now, every element $v\in V(G)$ orders its list containing the elements
of $N_r(v)\cap Z$ in an increasing order. We assume here, that every
vertex is identified by a number between $1$ and $n$. This takes 
time $\Oof(n\cdot |Z|\cdot \log |Z|)$. For two elements $v,w\in V(G)$, we can 
now decide whether $N_r(v)\cap Z=N_r(w)\cap Z$ in time $\Oof(|Z|)$
by comparing the ordered lists representing their neighbourhood 
intersections. 

We now order the vertices of $V(G)$ according to their neighbourhood
intersections in time $\Oof(n\cdot \log n\cdot |Z|)$. We remove
duplicates from the sorted list in time $\Oof(n\cdot |Z|)$ to output
the set $Y$. 

Note that $\log n\leq n^\epsilon$ for all fixed $\epsilon>0$ and 
sufficiently large $n$. By rescaling $\epsilon$, we get a total running
time of $\Oof(n^{1+\epsilon}\cdot |Z|^{1+\epsilon})$. 
\end{proof}

According to \cref{crl:neighbourhoodcomplexity}, the set
$Y$ we compute in \cref{lem:representatives} has size at most $\abs{Z}^s$
for some integer $s$ depending only on $r$ and $\CCC$. 

Let $Z$ be a $r$-domination core of $G$, let $Y\subseteq V(G)$ 
be a minimum size set such that for all $u\in V(G)$ there is $v\in Y$
with $N_r(u)\cap Z=N_r(v)\cap Z$. We construct a graph $H$ whose
vertex set is the union of $Z$ and $Y$. For every $v\in Y$, compute $N_r(v)\cap
Z$ and add the vertices and edges of a shortest path 
(of length at most $r$) between $v$ and
each $z\in N_r(v)\cap Z$ to $H$, such that
$N_r^G(v)\cap Z= N_r^H(v)\cap Z$. As $\abs{N_r^G(v)\cap Z}$ is 
bounded by $\abs{Z}$, we conclude that 
$\abs{V(H)}\leq r\cdot \abs{Z}^{s+1}$. We now add two new vertices 
$v,v'$ to $H$
and connect $v$ to every vertex except to the vertices 
of~$Z$ by a path of length $r$ and to~$v'$ by
a path of length $r$. It is easy to see (compare to Lemma 2.16 
of~\cite{drange2015kernelization}) that there exists a set $D$ of size
$k$ which $r$-dominates~$Z$ if, and only if, $H$ admits an distance-$r$ dominating 
set of size $k+1$. Hence, the size of $H$ is polynomially bounded by $k$ and 
$H$ can be computed in polynomial time by combining the above lemmas. 
This proves the main theorem of this section. 

\begin{theorem}
  Let $\CCC$ be a nowhere dense class of graphs. For every $r\in\N$ there is a
  polynomial~$p(x)$, a constant $c$ and an algorithm running in time
  $\Oof(p(k)^c\cdot n^{2+\epsilon})$ which, given a graph $G\in \CCC$, $\epsilon>0$
  and $k\in \N$ computes a
  kernel for the \textsc{Distance-$r$ Dominating Set} problem of size $p(k)$ on $G\in \CCC$.
\end{theorem}

\section{Single Exponential Parameterized Algorithms on Nowhere Dense Classes}
\label{sec:single}

In \cite{DawarK09}, the authors show how the concept of
uniform quasi-wideness can be used to design parameterized
algorithms for dominating set problems on nowhere dense classes of
graphs. However, the dependence on the parameter in the algorithms in  \cite{DawarK09}
is enormous, usually manyfold exponential. 

In this section we combine the tools developed in the previous
section with the general technique for obtaining parameterized
algorithms in
\cite{DawarK09} to design parameterized algorithms for several graph
problems on 
nowhere dense classes with only a single exponential dependence on the
parameter value.

We demonstrate the idea by showing that the \textsc{Connected
  Dominating Set} problem can be solved in time $2^{p(k)}\cdot n^c$ for a
fixed constant $c$ and a polynomial $p(x)$. This example is
particularly interesting as it was shown in
\cite{drange2016kernelization} that the problem is unlikely to have a
polynomial kernel on nowhere dense classes. Hence a single exponential
parameter dependence cannot be obtained from a polynomial
kernelization algorithm. 

\begin{theorem}\label{thm:single}
  Let $\CCC$ be a nowhere dense class of graphs. Then there is a
  polynomial $p(x)$, and an
  algorithm running in time $\Oof(2^{p(k)}\cdot n^{1+\epsilon})$ which, given an $n$-vertex
  graph $G$, $\epsilon>0$ and a number $k$ as input, decides whether $G$ contains a
  connected dominating set of size $k$. 
  \end{theorem}
\begin{proof}
  For $i\geq 1$ let $d_i$ be the minimum size of a clique that cannot be obtained as
  depth-$i$ minor in any member $H\in \CCC$. Let $s$ and $N$ be
  a margin of the class $\CCC$. By \cref{thm:uqw}, we can
  choose $s(r)$ such that $s(r)\le d_{r}$ and $N(r, m) \le
  m^{c(r)}$ for some function $c(r)$ depending only on
  $r$.

  Let $G$ and $k$ be given.   Let $s \coloneqq
  s(1) = d_1$ and let $K \coloneqq N(1, k+1)  = (k+1)^{c(1)}$. 
  During the algorithm we will maintain sets $W_i$ and $X_i$ of
  vertices such that $\abs{X_i} = i$ and $W_i$ is the set of vertices in
  $G$ not dominated by any member of $X_i$. 
  We initialise $W_0 \coloneqq V(G)$ and $X_0 \coloneqq \emptyset$. 

  Now, after $i$ steps, suppose that $W_i, X_i$ have been defined. 
  If $\abs{W_i}\geq K$, then we use \cref{thm:uqw} to compute a set
  $S$ of size $\abs{S}\leq s$ and a set $A\subseteq W$ of size $k+1$ such
  that $A$ is $2$-independent in $G-S$. As $A$ is $2$-independent, no
  vertex in $G-S$ can dominate two members of $A$. Hence, if there is
  a connected dominating set in $G$ using the vertices in $X_i$, then
  this set will need to include an element of $S$. We now branch over
  the $s$ possible choices of an element in $S$ and for each such
  $v\in S$ call the algorithm recursively with sets $X_{i+1} \coloneqq
  X_i\cup \{ v\}$ and $W_{i+1} \coloneqq W_i\setminus N[v]$, where $N[v]$
  denotes the closed neighbourhood of $v$, i.e.~$N[v] \coloneqq N(v) \cup \{v\}$. 
  
  If $i=k$, the recursion stops. If $W\neq \emptyset$, then we can return
  ``no'' as there cannot be any connected dominating set of size at most $k$
  using the vertices in $X_i$. If $W_i = \emptyset$, then we found a solution if,
  and only if, $X_i$ is connected.

  Otherwise we have $i<k$ and the recursion stops because $\abs{W_i}< K$.
  We now have a set $X_i$ of
  $\abs{X_i} = i < k$ vertices for our dominating set and still need to
  dominate $W_i$. Furthermore, we still need to connect the dominating set.

  We suppose that $\ell\le k-i$ vertices $y_1,\ldots,y_\ell$ are used to dominate
    $W_i$ and at most $k-i-\ell$ further vertices to connect the dominating set $X_i\cup\{y_1,\ldots,y_\ell\}$. For
    every $y_j$ we guess a set $Y_j$ that is dominated by $y_j$ and such that the sets
  $Y_j$ form a partition of $W_i$. We do not forbid that a vertex~$y_j$ also
  dominates some vertices in a $Y_{j'}$ for $j\neq j'$. Here the $Y_j$'s are guessed
and the $y_j$'s  are then computed. In other words, 
  for any partition of $W_i$ into $l\leq k-i$ non-empty sets $Y_1, \dots Y_l$ we
  do the following. First, for each $Y_j$ we compute the set $D_j$ of vertices
  $v\in V(G)$ such that $Y_j \subseteq N[v]$.  If for some $j$ there is no such
  vertex, then we discard this partition. So suppose $Y_1, \dots, Y_l$ is a
  partition such that $D_1, \dots, D_l$ are all non-empty. Then we can define a
  dominating set of $G$ by adding one element $w_j$ of each $D_j$ to the set
  $X_i$. By construction, every vertex $v\in V(G)\setminus W_i$ is dominated by
  a member of $X_i$ and every $v\in W_i$ is dominated by the vertex $w_j$ chosen
  for the partition $Y_j$ containing $v$.

  Hence, all that remains is to show that we can choose the $w_1,
  \dots, w_l$ so that $X_i \cup \{ w_1, \dots, w_l\}$ can be turned
  into a connected set by adding at most $k-(i+l)$ extra vertices.
  We solve this problem by using the Dreyfus-Wagner algorithm
  \cite{DreyfusW72} for
  solving Steiner trees. The Dreyfus-Wagner algorithm computes a
   minimum Steiner tree for a set of at most $T$ terminals in
  $\Oof(3^Tn + 2^Tn^2 + n^2 \log n + nm)$ time, where $m$ is the
  number of edges of the input graph  on $n$ vertices.

  Now, suppose we are given the set $X_i$ of size $i$, the partition $Y_1, \dots, Y_l$ of $W_i$
  into disjoint sets, for $l\leq k-i$ and the sets $D_1, \dots,
  D_l$. For all $1\leq j \leq l$ we add a fresh vertex $t_j$ to the
  graph~$G$ and add a path $P(t_j, v)$ of length $3$ between $t_j$
  and every member $v$ of
  $D_j$ (so that the paths $P(t_j, v)$ and $P(t_{j'}, v')$ are
  internally vertex disjoint if $\{t_j, v\}\neq \{t_{j'}, v'\}$). 
  Let $G'$ be the augmented graph. We now call the
  Dreyfus-Wagner algorithm on $G'$ with terminal set $X_i \cup \{
  t_1, \dots, t_l\}$. Let $T$ be the resulting Steiner tree. Then $T$
  needs to contain at least one vertex of every set $D_j$, as this is
  the only way to connect the terminal $t_j$ to the rest of the
  graph. Hence, any Steiner tree~$T$ is a connected dominating set of
  $G'$. 

  Note that in any minimum size Steiner tree for this terminal
  set, every $t_j$ is connected by exactly one path $P(t_j, v)$ to a
  vertex $v\in D_j$. To see this, recall that every $w\in Y_j$ has
  every vertex in $D_j$ as its neighbour. The only reason why a
  minimum Steiner tree might contain two paths $P(t_j, v)$ and $P(t_j, v')$
  for distinct $v, v'\in D_j$ is that this is needed to connect $v$
  and $v'$. But then, the path $P(t_j, v')$ can be replaced by adding
  any member of $Y_j$ instead. As $P(t_j, v')$ has two internal
  vertices, this would decrease the size of the Steiner tree.
  In particular this implies that every vertex $t_j$ is a leaf of a
  minimum size Steiner tree.
  
  We can now turn $T$ into a connected dominating set for
  $G$ as
  follows. We simply delete for each~$t_j$ the vertex $t_j$ and all
  internal vertices of the (unique) path $P(t_j, v)$ connecting $t_j$
  to a member of $D_j$. By the argument above, the resulting tree $T'$ is still
  connected and forms a connected dominating set of $G$. Conversely,
  any minimum size connected dominating set for~$G$ containing $X_i$
  and also at least one vertex of each $D_j$ can be extended to a
  minimum size Steiner tree in $G'$ with the terminal set as above. It
  follows that if for any partition $Y_1, \dots, Y_t$ we obtain a
  Steiner tree as above of size at most~$k$, then we can return this
  tree as a connected dominating set, and otherwise can conclude that
  there is no connected dominating set of size at most $k$ using the
  vertices in $X_i$ and the partition $Y_1,\ldots,Y_\ell$. 

  Hence, the whole search tree yields a correct decision procedure for the
  connected dominating set problem. The running time is bounded by the size of
  the search tree and the time the algorithm takes at each leaf of the search
  tree.  Note that the branching of the search tree is bounded by $s = d_1$ and
  the depth by $k$. Hence, in total we have at most $s^{k+1}-1$ nodes. At every
  leaf we have to consider every possible partition of the set $W_i$ of size at
  most $K = (k+1)^{c(1)}$ into at most $k$ disjoint sets. For each partition we
  call the Dreyfus-Wagner algorithm whose running time is dominated by
  $\Oof(3^k \abs{W_i}^2 \cdot \log\abs{W_i})$.  Hence, the total running time of the algorithm is
  bounded by $2^{p(k)}\cdot n^{1+\epsilon}$ for a polynomial $p$.
\end{proof}

 The proof method used in the previous theorem can be used to
 establish single exponential parameterized algorithms for other
 problems on nowhere dense classes of graphs. Hence, while the general
 proof technique is the same, we can dramatically improve the running
 time of our algorithms from a multiple exponential parameter
 dependence to single exponential.

\section{Conclusion and recent results}

Nowhere dense classes are classes of uniformly sparse graphs with
a rich algorithmic theory. Especially the characterisation of these
classes via uniform quasi-wideness finds many algorithmic applications. 
In this paper we have proved polynomial bounds on the margin of
uniformly quasi-wide classes, and hence nowhere dense classes. 
These bounds yield new algorithmic tools, especially in the area of parameterized
complexity for nowhere dense graph classes. We have derived a 
polynomial kernel for the \textsc{Distance-$r$ Dominating Set} problem
and a parameterized algorithm with single exponential parameter
dependence for the \textsc{Connected Dominating Set} problem. 

The polynomial bounds on the margin of uniformly quasi-wide classes
are established by a connection to stability theory, a classical field 
of infinite model theory. The proof is based a non-constructive argument of 
Adler and Adler~\cite{adler2014interpreting}, who observed that
nowhere dense graph classes are stable. As a consequence, our 
upper bounds are given purely existentially and are not effectively
computable. A purely combinatorial and effective proof has recently
been given in~\cite{pilipczuk2018number}. Lower bounds, showing that
polynomial bounds are the best one can hope for even on classes of 
graphs which exclude a fixed minor, were established in~\cite{NadaraPRRS18}.

The kernelization results for the \textsc{Distance-$r$ Dominating Set}
problem have recently been improved from polynomial kernels to almost
linear kernels in~\cite{EickmeyerGKKPRS17}. The \textsc{Connected
Dominating Set} problem does not admit polynomial kernels even 
on bounded expansion classes of graphs (unless $\text{NP}\subseteq \text{coNP/poly}$)~\cite{drange2016kernelization}. However, 
lossy kernels of linear size for the 
\textsc{Connected Dominating Set} problem 
on classes of bounded expansion and lossy kernels of
polynomial size for the
\textsc{Connected Distance-$r$ Dominating Set} problem 
on nowhere dense classes were 
recently established in~\cite{EibenKMPS18}.

\bibliographystyle{plain}
\bibliography{ref}

\end{document}